\documentclass[11pt,fleqn]{article}
\usepackage{amsmath,amssymb,amsthm,bbm}
\usepackage{xcolor}
\usepackage{booktabs}
\usepackage{graphicx}
\usepackage{enumitem}
\usepackage[left=2cm,right=2cm,top=2cm,bottom=2cm]{geometry}
\usepackage[pdfstartview=FitH,bookmarksopen=false,
colorlinks=true,linkcolor=blue,citecolor=blue]{hyperref}

\newtheorem{theorem}{Theorem}[section]
\newtheorem{proposition}[theorem]{Proposition}

\theoremstyle{definition}

\newtheorem{remark}[theorem]{Remark}

\newcommand{\R}{\mathbb{R}}

\newcommand{\abs}[1]{\lvert #1 \rvert}

\title{\bf Three characterizations of the weighted center of imputations value\thanks{This work was supported by the National Natural Science Foundation of China (No.~72371151)}}
\author{Erfang Shan$\,^1$\thanks{E-mail addresses: \emph{efshan@shu.edu.cn} (E.~Shan), \emph{lykang@shu.edu.cn} (L.~Kang)},\, Liying Kang$^{\,2}$\\
{\small $^{1}$School of Management, Shanghai University, Shanghai 200444, P.R.~China}\\
{\small $^2$Department of Mathematics,  Shanghai University, Shanghai 200444, P.R.~China}}
\date{}

\begin{document}
\maketitle

\begin{abstract}
The weighted center of imputations (CIS) value allocates the surplus of the grand coalition equally after granting each player a fixed proportion of his individual worth. This paper provides three axiomatic characterizations of this value by generalizing the individual rationality and subgame order preservation axioms. The first characterization employs individual rationality with respect to the weights together with the equal surplus increment property. The second relies on efficiency, additivity, symmetry adjusted by the weights, and a dummifying player property adapted to the weights. The third builds on efficiency and a weak subgame order preservation axiom that incorporates the weights. These results unify and extend recent findings, covering both the equal division and the standard CIS values as special cases.
\bigskip
		
		
\medskip
		\noindent {\bf AMS (2000) subject classification:} 91A12
		
\end{abstract}

\section{Introduction}

A transferable utility game (TU-game) is a pair $(N,v)$, where $N$ is a finite set of players and $v:2^N\to\R$ satisfies $v(\emptyset)=0$.
Among the many solution concepts proposed, those that combine respect for individual contributions with egalitarian sharing have attracted considerable interest.
The center of imputations (CIS) value\footnote{This value is nowadays also often called the \emph{equal surplus division value} (e.g., van den Brink, 2007; Casajus and Huettner, 2014).}, introduced by Driessen and Funaki \cite{Dr_1991}, first guarantees each player his individual worth and then divides the remaining surplus equally among all players.
A natural weighted extension, in which each player may claim only a fraction of his individual worth before egalitarian sharing, was introduced by Hou et al. \cite{Hou2024} under the name \emph{generalized CIS value} (or \emph{weighted CIS value}).
Formally, for a given vector of weights $\alpha=(\alpha_i)_{i\in N}$ with $\alpha_i\in[0,1]$, the weighted CIS value is
\begin{equation}\label{eq:weightedcis}
CIS^\alpha_i(N,v)=\alpha_i v(\{i\})+\frac{1}{\abs{N}}\Bigl(v(N)-\sum_{j\in N}\alpha_j v(\{j\})\Bigr).
\end{equation}
The parameter \(\alpha_i\) represents the selfishness level of player $i$, with distinct implications in different game settings. For cost games, a lower \(\alpha_i\) means the player is more selfish and reluctant to take on individual costs. For surplus-sharing (payoff) games, a higher
\(\alpha_i\) indicates greater selfishness, as the player pursues a larger portion of standalone payoffs.
In particular, when all $\alpha_i=0$, we obtain the equal division value $ED_i(N,v)=v(N)/|N|$; when all $\alpha_i=1$, we recover the standard CIS value; and when all $\alpha_i$ equal a common constant $\alpha$, we obtain the $\alpha$-CIS value, which can be expressed as a convex combination of the equal division value and the CIS value (see van den Brink et al. \cite{vanBrinkChunFunaki2016}; Xu et al., \cite{XuDaiShi2015}).
Thus the weighted CIS value interpolates flexibly between purely egalitarian and purely individualistic benchmarks.

In this paper we advance this line of research by generalizing standard axioms, such as the individual rationality, symmetry, the dummifying player property (\cite{CasajusHuettner2014}), and subgame order preservation axiom (\cite{Navarro2025}), to the weighted setting, then develop three distinct axiomatic characterizations for the weighted CIS value.
First, we propose a characterization founded on the generalized individual rationality (namely  $\alpha$-individual rationality) and the equal surplus increment property, which generalizes the characterization for the standard CIS value established by Calleja and Llerena \cite{CallejaLlerena2017}.
Second, we  construct another characterization based on four axioms: efficiency, additivity,
generalized symmetry ($\alpha$-symmetry), and the generalized dummifying player property ($\alpha$-dummifying player property).
This characterization unifies and extends the early works of van den Brink \cite{vanBrink2007} on the equal division value and Casajus and Huettner \cite{CasajusHuettner2014} on the CIS value. Specifically, van den Brink \cite{vanBrink2007} introduced the nullifying player property to characterize the equal division value, while Casajus and Huettner \cite{CasajusHuettner2014} used the dummifying player property to characterize the CIS value. By replacing these with the $\alpha$-dummifying player property, our second result covers both extremes and all intermediate cases in a single framework.

Finally, we derive a characterization via  efficiency and weak subgame order preservation with respect to the weighted vector $\alpha$ (namely $\alpha$-subgame sum-order preservation). This weakened axiom revises Navarro's original subgame order preservation axiom in two respects: it replaces the original worths by $\alpha$-adjusted worths and requires only the sum of payoffs over the intersection to be ordered.
Consequently, this result integrates the order-preserving characterizations of the equal division value and the CIS value proposed by Navarro \cite{Navarro2025} and Lowing, Nakada, and Navarro \cite{Lo2026} into this generalized weighted research framework.

The remainder of the paper is organized as follows. Section~\ref{sec:prel} introduces basic definitions and axioms. Section~\ref{sec3}
presents the three characterizations for the weighted CIS value, and concluding remarks are given in Section \ref{sec:concl}.

\section{Preliminaries}\label{sec:prel}

Let $U$ be an infinite universe of potential players.
For a finite set $N\subset U$, a \emph{TU-game} is a pair $(N,v)$ with $v:2^N\to\R$ and $v(\emptyset)=0$.
Subsets of $N$ are called
{\em coalitions}, and $v(S)$ is the {\em worth} of coalition $S$ generated by cooperation.
The class of all TU-games is denoted by $\mathcal{G}$.

For $S\subseteq N$, $\abs{S}$ denotes its cardinality.
For $x\in\R^N$ and $S\subseteq N$, write $x(S)=\sum_{i\in S}x_i$.
For a nonempty coalition $S\subseteq N$, the \emph{subgame} on $S$ is the game $(S,v|_S)$ where $v|_S(T)=v(T)$ for all $T\subseteq S$.

The set of \emph{feasible payoff vectors} of $(N,v)$ is
\[
X^*(N,v)=\{x\in\R^N\mid x(N)\le v(N)\},
\]
and the set of \emph{pre-imputations} is
\[
X(N,v)=\{x\in\R^N\mid x(N)=v(N)\}.
\]
A \emph{value} is a mapping $\psi$ that assigns to every $(N,v)\in\cal{G}$ a feasible payoff vector $\psi(N,v)\in X^*(N,v)$.

Fix a vector of weights $\alpha=(\alpha_i)_{i\in U}$ with $\alpha_i\in[0,1]$, independent of the game.
For a finite player set $N\subset U$ we write $\alpha$ for its restriction.
The \emph{weighted CIS value} $CIS^\alpha$ is given by \eqref{eq:weightedcis}.
Following Hou et al. \cite{Hou2024} (who studied cost games), $\alpha_i$ measures {\em selfishness}: smaller $\alpha_i$ means more selfish in cost games; in surplus-sharing games, the opposite holds. We therefore treat $\alpha_i$ as a weight without loss of generality.

Define the \emph{$\alpha$-residual game} $v^\alpha$ by
\begin{equation}\label{eq:valpha}
v^\alpha(S)=v(S)-\sum_{j\in S}\alpha_j v(\{j\}),\qquad S\subseteq N.
\end{equation}
Then $CIS^\alpha_i(N,v)=\alpha_i v(\{i\})+\frac{1}{\abs{N}}v^\alpha(N)$.

For a vector $\alpha\in[0,1]^N$, we generalize the classical notions of imputation and essential games as follows. When all $\alpha_i = 1$, these definitions reduce to the standard ones.
The \emph{$\alpha$-imputation set} is defined by
\[
I^\alpha(N,v)=\left\{x\in X(N,v)\;\middle|\;x_i\ge \alpha_i v(\{i\}),\ \forall i\in N\right\}.
\]
A game is called \emph{$\alpha$-essential} if $I^\alpha(N,v)\neq\emptyset$, i.e.,
\[
v(N)\ge \sum_{i\in N}\alpha_i v(\{i\}).
\]


 A player $i\in N$ is  a \emph{nullifying player} in game $(N,v)$ (see \cite{vanBrink2007}) if $v(S)=0$ for all $S\ni i$.
A player $i\in N$ is  a \emph{dummifying player} (see \cite{CasajusHuettner2014}) if $v(S)=\sum_{j\in S}v(\{j\})$ for all $S\ni i$.

We unify and generalize the above two notions as follows.

A player $i\in N$ is called an \emph{$\alpha$-dummifying player} in $(N,v)$ if for all $S\subseteq N$ with $i\in S$,
\[
v(S)=\sum_{j\in S}\alpha_j v(\{j\}), i.e.,  v^\alpha(S)=0.
\]

Obviously, if $\alpha_j=0$ for all $j$, the $\alpha$-dummifying player reduces to a nullifying player; if $\alpha_j=1$ for all $j$, it coincides with a dummifying player.

For any nonempty $T\subseteq N$, the \emph{unanimity game} $u_T$ is defined by $u_T(S)=1$ if $T\subseteq S$, and $0$ otherwise.
Every game $v$ can be uniquely expressed as a linear combination of unanimity games,
\[
v=\sum_{\emptyset\neq T\subseteq N} c_T u_T,
\]
where the coefficients
\[
c_T=\sum_{S\subseteq T}(-1)^{|T|-|S|}\,v(S)
\]
are the \emph{Harsanyi dividends} of coalition $T$ (see \cite{Har_1959}).

We now list the axioms used in our characterizations. Those involving $\alpha$ generalize or unify  classical properties, including individual rationality, symmetry, dummifying (nullifying) player property, and subgame order preservation, to the weighted setting.

{\bf Efficiency}, \textbf{E}. \label{ax:e}
 For all $(N,v)\in\cal{G}$, $\sum_{i\in N}\psi_i(N,v)=v(N)$.

{\bf Additivity}, \textbf{A}.\label{ax:a}  For any two games $(N,v), (N,w)\in \cal{G}$,
$\psi(N,v+w)=\psi(N,v)+\psi(N,w)$.

{\bf Equal Surplus Increment}\footnote{This axiom is called  {\em equal surplus division} in \cite{CallejaLlerena2017}. To distinguish it from the equal surplus division value, we refer to it as the equal surplus increment property.}, \textbf{ESI}. \label{ax:esi}
For any two games $(N,v), (N,v')\in \cal{G}$ such that $v(S)=v'(S)$ for all $S\subset N$ and $v'(N)=v(N)+t$ for some $t\in\R$,
\[
\psi_i(N,v')-\psi_i(N,v)=\frac{t}{\abs{N}}\quad\text{for all } i\in N.
\]

{\bf $\alpha$-Individual Rationality}, $\mathbf{IR^\alpha}$. \label{ax:ira}
For all $\alpha$-essential game $(N,v)$,
\[
\psi_i(N,v)\ge \alpha_i v(\{i\})\quad\text{for all } i\in N.
\]

When all $\alpha_i=1$, this reduces to the classical individual rationality: $\psi_i(N,v)\ge v(\{i\})$.

{\bf $\alpha$-Symmetry}, \textbf{$\alpha$-S}. \label{ax:gs}
If players $i,j\in N$ are symmetric in $(N,v)$ (i.e., $v(S\cup\{i\})=v(S\cup\{j\})$ for all $S\subseteq N\setminus\{i,j\}$), then
\[
\psi_i(N,v)-\alpha_i v(\{i\})=\psi_j(N,v)-\alpha_j v(\{j\}).
\]

When $\alpha_i=\alpha$ for all $i\in N$, \textbf{$\alpha$-S} reduces to the standard symmetry axiom requiring that symmetric players receive equal payoffs after deducting their weighted individual worths.

{\bf $\alpha$-Dummifying Player Property}, \textbf{$\alpha$-DP}. \label{ax:adp}
If player $i$ is an $\alpha$-dummifying player in $(N,v)$, then $\psi_i(N,v)=\alpha_i v(\{i\})$.

This property unifies the nullifying player property (when $\alpha_i=0$, see \cite{vanBrink2007})  and the dummifying player property (when $\alpha_i=1$, see \cite{CasajusHuettner2014}).

For completeness, we recall two axioms introduced by Navarros \cite{Navarro2025} and Lowing, Nakada, and Navarro \cite{Lo2026}, which serve as benchmarks for our weakened and unified axiom.

{\bf Subgame Order Preservation}, \textbf{SOP}. \label{ax:sop}
For every $(N,v)\in\cal{G}$, and for any two coalitions $S,T\subseteq N$ with $\abs{S}=\abs{T}$ and $P=S\cap T\neq\emptyset$, if $v(S)\ge v(T)$ then $\psi_i(S,v|_S)\ge\psi_i(T,v|_T)$ for all $i\in P$.

{\bf Subgame order preservation w.r.t. coalitional surplus}, \textbf{SOPCs}.\label{ax:sopc}
For every $(N,v)\in\cal{G}$, and for any two coalitions $S,T\subseteq N$ with $\abs{S}=\abs{T}$ and $P=S\cap T\neq\emptyset$, if $v(S)-\sum_{j\in S}v(\{j\})\ge v(T)-\sum_{j\in T}v(\{j\})$\footnote{This inequality in \cite{Lo2026}  is written as: $v(S)-\sum_{j\in S\setminus P}v(\{j\})\ge v(T)-\sum_{j\in T\setminus P}v(\{j\})$, and is equivalent to our expression here.}, then $\psi_i(S,v|_S)\ge\psi_i(T,v|_T)$ for all $i\in P$.

We now propose the following axiom, which weakens and unifies {\bf SOP} and {\bf SOPCs} into a single general formulation.

{\bf $\alpha$-Subgame Sum-Order Preservation}, \textbf{$\alpha$-SSOP}. \label{ax:wsp}
For every $(N,v)\in\cal{G}$, and for any two coalitions $S,T\subseteq N$ with $\abs{S}=\abs{T}$, $P=S\cap T\neq\emptyset$, if $v^\alpha(S)\ge v^\alpha(T)$, then
\[
\sum_{i\in P}\psi_i(S,v|_S)\ge\sum_{i\in P}\psi_i(T,v|_T).
\]

The axiom \textbf{$\alpha$-SSOP} generalizes and unifies  \textbf{SOP} and {\bf SOPCs} in two respects: it uses the $\alpha$-adjusted worths $v^\alpha$  instead of
 $v$,  and it only requires the sum of payoffs over the intersection to be ordered, instead of pointwise comparisons.
Moreover, this axiom reduces to weak versions of \textbf{SOP} and  {\bf SOPCs} when all $\alpha_i=0$ and all $\alpha_i=1$, respectively.

\section{Axiomatizations}\label{sec3}
This section provides three axiomatic characterizations of the weighted CIS value.

\subsection{Characterization by  $\alpha$-Individual Rationality  and Equal Surplus Increment
}\label{sec:IRESI}

We begin with the characterization that combines $\alpha$-individual rationality  and the equal surplus increment property.

\begin{theorem}\label{thm:IRESI}
A value $\psi$ on $\mathcal{G}$ satisfies $\alpha$-individual rationality  (\textbf{IR}$^\alpha$) and equal surplus increment (\textbf{ESI}) if and only if $\psi=CIS^\alpha$.
\end{theorem}

\begin{proof}
\textbf{Existence.} For an $\alpha$-essential game, since $v^\alpha(N)\ge 0$, $CIS^\alpha_i(N,v)=\alpha_i v(\{i\})+v^\alpha(N)/\abs{N}\ge\alpha_i v(\{i\})$, thus \textbf{IR}$^\alpha$ holds. For \textbf{ESI}, if $v'$ differs from $v$ only on $N$ by $t$, then $(v')^\alpha(N)=v^\alpha(N)+t$, giving $CIS^\alpha_i(N,v')=CIS^\alpha_i(N,v)+t/\abs{N}$.

\medskip
\textbf{Uniqueness.} Let $\psi$ satisfy \textbf{IR}$^\alpha$ and \textbf{ESI}. For an arbitrary game $(N,v)$, construct the auxiliary game $(N,v_\alpha^I)$ by
\[
v_\alpha^I(S)=v(S)\quad\text{for all } S\subset N,\qquad
v_\alpha^I(N)=\sum_{i\in N}\alpha_i v(\{i\}).
\]
By construction $v_\alpha^I(N)=\sum_{i\in N}\alpha_i v(\{i\})$, so $(N,v_\alpha^I)$ is $\alpha$-essential.
Since $\psi$ is feasible,
\[\sum_{i\in N}\psi_i(N,v_\alpha^I)\le v_\alpha^I(N)=\sum_i\alpha_i v(\{i\}).\]
On the other hand, \textbf{IR}$^\alpha$ yields $\psi_i(N,v_\alpha^I)\ge \alpha_i v(\{i\})$ for each $i$, hence we  have $\psi_i(N,v_\alpha^I)=\alpha_i v(\{i\})$ for all $i\in N$.

Now set $d=v(N)-\sum_{j\in N}\alpha_j v(\{j\})$. The games $(N,v_\alpha^I)$ and $(N,v)$ coincide on every proper coalition $S\subset N$, and $v(N)=v_\alpha^I(N)+d$. Applying \textbf{ESI},
\[
\psi_i(N,v)=\psi_i(N,v_\alpha^I)+\frac{d}{\abs{N}}=\alpha_i v(\{i\})+\frac{1}{\abs{N}}\Bigl(v(N)-\sum_{j\in N}\alpha_j v(\{j\})\Bigr)=CIS^\alpha_i(N,v).
\]
Thus $\psi=CIS^\alpha$.
\end{proof}

\subsection{Characterization via Additivity, $\alpha$-Symmetry, and the $\alpha$-Dummifying Player Property}\label{sec:add}

We now present the second characterization, which relies on additivity, $\alpha$-symmetry, and the $\alpha$-dummifying player property.

\begin{theorem}\label{thm:add}
A value $\psi$ on $\mathcal{G}$ satisfies efficiency (\textbf{E}), additivity (\textbf{A}), $\alpha$-symmetry (\textbf{$\alpha$-S}), and the $\alpha$-dummifying player property (\textbf{$\alpha$-DP}) if and only if $\psi=CIS^\alpha$.
\end{theorem}

\begin{proof}
\textbf{Existence.} One can easily check that the $CIS^\alpha$ value satisfies the four axioms via direct substitution.

\textbf{Uniqueness.} Fix $N$ with $\abs{N}=n$. By \textbf{A}, it suffices to show $\psi=CIS^\alpha$ on scaled unanimity games $c u_T$ ($c\in\R$, $\emptyset\neq T\subseteq N$).

\emph{Case 1: $\abs{T}\ge 2$.} In $c u_T$, all singletons have value zero. Players in $T$ are symmetric, as are players outside $T$. By \textbf{$\alpha$-S} there exist $a_t,b_t$ with $\psi_i(c u_T)=a_t$ for $i\in T$, $\psi_k(c u_T)=b_t$ for $k\notin T$, and \textbf{E} gives $t a_t+(n-t)b_t=c$, where $t=\abs{T}$.

Pick $k\notin T$ and define $w=c u_T-c u_{T\cup\{k\}}$. In $w$, all singleton coalitions are zero and $w(S)=0$ for every $S\ni k$; thus $k$ is $\alpha$-dummifying. By \textbf{$\alpha$-DP}, $\psi_k(w)=0$. \textbf{A} yields $\psi_k(w)=\psi_k(c u_T)-\psi_k(c u_{T\cup\{k\}})$. In $c u_{T\cup\{k\}}$, player $k$ belongs to the coalition, so its payoff is the internal value for size $t+1$, denoted $a_{t+1}$. Hence $b_t=a_{t+1}$.

When $T=N$, \textbf{$\alpha$-S} forces equality $\psi_i(c u_N)=a_n$ for all $i\in N$, and  \textbf{E} leads to $a_n=c/n$. Applying $b_t=a_{t+1}$ backwards from $t=n-1$ down to $t=2$ and the  equations gives $a_t=b_t=c/n$ for all $t\ge 2$. Hence $\psi_i(c u_T)=c/n$ for all $i\in N$. Note that $CIS^\alpha_i(cu_T)=c/n$, so $\psi(cu_T)=CIS^\alpha(cu_T)$.

\emph{Case 2: $T=\{i\}$.} In $c u_{\{i\}}$, $c u_{\{i\}}(\{i\})=c$, $c u_{\{i\}}(\{j\})=0$ for $j\neq i$. All $j\neq i$ are symmetric, so by \textbf{$\alpha$-S}, $\psi_i(c u_{\{i\}})=x$, $\psi_j(c u_{\{i\}})=y$ for $j\neq i$, and \textbf{E} gives $x+(n-1)y=c$.

Fix $j\neq i$ and set $w=c u_{\{i\}}+(\alpha_i-1)c u_{\{i,j\}}$.
For singleton coalitions: $w(\{i\})=c$, $w(\{j\})=0$, and $w(\{k\})=0$ for all $k\in N\setminus\{i,j\}$. For any $S\ni j$, if $i\notin S$, then clearly $w(S)=\sum_{\ell\in S}\alpha_\ell w(\{\ell\})=0$; if $i\in S$,
\[
w(S)=c+(\alpha_i-1)c=\alpha_i c=\sum_{\ell\in S}\alpha_\ell w(\{\ell\}),
\]
because $w(\{\ell\})=0$ for $\ell\neq i$. Hence $j$ is $\alpha$-dummifying. By \textbf{$\alpha$-DP}, $\psi_j(w)=0$. \textbf{A}  and Case~1 give
\[
0=\psi_j(w)=y+(\alpha_i-1)\frac{c}{n},
\]
so $y=c(1-\alpha_i)/n$ and $x=\alpha_i c+c(1-\alpha_i)/n$. Note that $CIS^\alpha_i(cu_i)=\alpha_i c+c(1-\alpha_i)/n=x$ and $CIS^\alpha_j(cu_{\{i\}})=c(1-\alpha_i)/n=y$ for all $j\in N\setminus\{i\}$,
so $\psi(c u_{\{i\}})=CIS^\alpha(cu_{\{i\}})$ for all $i\in N$.

Now for arbitrary $(N,v)$, write $v=\sum_{T\subseteq N} c_T u_T$. {\bf A} gives $\psi(v)=\sum_T\psi(c_T u_T)$.
In the previous steps, we have shown that $\psi(c_T u_T) = CIS^\alpha(c_T u_T)$ for every term $c_T u_T$. Substituting this into the sum gives
\[
\psi(v) = \sum_{T} CIS^\alpha(c_T u_T).
\]
Because the weighted CIS value $CIS^\alpha$ itself satisfies additivity ({\bf A}), the right-hand side equals $CIS^\alpha\!\left(\sum_{T} c_T u_T\right)$, which is precisely $CIS^\alpha(v)$ by the Harsanyi decomposition. Hence $\psi(v) = CIS^\alpha(v)$ for all games $v$, completing the proof.
\end{proof}

The four axioms in Theorem~\ref{thm:add} are logically independent. To see this, we construct four values,  each satisfying exactly three of the axioms \textbf{E}, \textbf{A}, \textbf{$\alpha$-S}, and \textbf{$\alpha$-DP}.
\begin{enumerate}

\item For all $(N,v)$, define
\[
\psi^1_i(N,v)=CIS^\alpha_i(N,v)+v(\{i\})-\frac{1}{|N|}\sum_{k\in N}v(\{k\}).
\]
This value satisfies \textbf{E}, \textbf{A}, \textbf{$\alpha$-S}, but does \emph{not} satisfy \textbf{$\alpha$-DP}.

\item Fix two distinct players $1$ and $2$. For any game $(N,v)$ with $\{1,2\}\subseteq N$, define
\[
\psi^2_i(N,v)=CIS^\alpha_i(N,v)+\bigl(v^\alpha(\{1\})-v^\alpha(\{2\})\bigr)\times
\begin{cases}
+1 & \text{if } i=1,\\
-1 & \text{if } i=2,\\
0 & \text{otherwise},
\end{cases}
\]
and $\psi^2(N,v)=CIS^\alpha(N,v)$ otherwise.
This value satisfies \textbf{E}, \textbf{A}, \textbf{$\alpha$-DP}, but does \emph{not} satisfy \textbf{$\alpha$-S}.

\item For all $(N,v)$, define
\[
\psi^3_i(N,v)=
\begin{cases}
\alpha_i v(\{i\})+\dfrac{v(\{i\})}{\sum_{k\in N}v(\{k\})}\Bigl(v(N)-\sum_{j\in N}\alpha_j v(\{j\})\Bigr) & \text{if } \sum_{k}v(\{k\})>0,\\[8pt]
CIS^\alpha_i(N,v) & \text{if } \sum_{k}v(\{k\})=0.
\end{cases}
\]
This value satisfies \textbf{E}, \textbf{$\alpha$-S}, \textbf{$\alpha$-DP}, but does \emph{not} satisfy \textbf{A}.

\item For all $(N,v)$, define
\[
\psi^4_i(N,v)=\alpha_i v(\{i\}).
\]
This value satisfies \textbf{A}, \textbf{$\alpha$-S}, \textbf{$\alpha$-DP}, but does \emph{not} satisfy \textbf{E}.

\end{enumerate}

These four values show that no axiom is redundant in the characterization.
\subsection{Characterization by Efficiency and $\alpha$-Subgame Sum-Order Preservation}\label{sec:WSP}

This section provides an axiomatic characterization of the weighted CIS value using only efficiency and the $\alpha$-subgame sum-order preservation axiom. The key step is to show that these two axioms imply $\alpha$-symmetry.

\begin{proposition}\label{prop:sym}
Let $\psi$ satisfy efficiency (\textbf{E}) and $\alpha$-subgame sum-order preservation (\textbf{$\alpha$-SSOP}). Then, for any game $(N,v)$ with $|N|\ge 2$,
\[
\psi_{i}(N,v)-\alpha_{i}v(\{i\})=\psi_{j}(N,v)-\alpha_{j}v(\{j\}),\qquad\forall i,j\in N.
\]
In particular, $\psi$ satisfies $\alpha$-symmetry (\textbf{$\alpha$-S}).
\end{proposition}

\begin{proof}
Let $n=|N|\ge 3$.
Take a set $K$ of $n-1$ new players disjoint from $N$.
For each $i\in N$, set $S_{i}=K\cup\{i\}$; then $|S_{i}|=n$.
Construct an auxiliary game $w$ on $N\cup K$ as follows:
\begin{itemize}
  \item $w|_N = v$, i.e., $w(S)=v(S)$ for all $S\subseteq N$.
  \item For each $k\in K$, $\alpha_k=0$, $w(\{k\})=0$, and $w(R)=0$ for all $R\subseteq K$.
  \item For each $i\in N$, $w(S_{i}) = v^{\alpha}(N) + \alpha_{i}v(\{i\})$.
  \item For any $k\in K$ and $i\in N$, $w(S_{i}\setminus\{k\}) = w(S_{i})$.
  \item All other coalitions have worth $0$.
\end{itemize}
One checks that $w^{\alpha}(S_{i}) = v^{\alpha}(N)$ for all $i\in N$.

Apply \textbf{$\alpha$-SSOP} to  $(N, w|_N)$ and $(S_{i},w|_{S_{i}})$.
Note that $|S_i|=|N|=n$, $S_i\cap N=\{i\}$, and $v^{\alpha}(N)=w^{\alpha}(S_{i})$.
Hence,
\[
\psi_{i}(N,v)= \psi_{i}(N,w|_N) = \psi_{i}(S_{i},w|_{S_{i}}). \tag{1}
\]
Next, compare $S_{i}$ and $S_{j}$ for $i,j\in N$.
Their intersection is $K$, and $w^{\alpha}(S_{i})=w^{\alpha}(S_{j})$.
Thus, by \textbf{$\alpha$-SSOP}, we have
\[
\sum_{k\in K}\psi_{k}(S_{i},w|_{S_{i}}) = \sum_{k\in K}\psi_{k}(S_{j},w|_{S_{j}}). \tag{2}
\]

The {\bf E} axiom on $(S_{i},w|_{S_{i}})$ gives
\[
\psi_{i}(S_{i},w|_{S_{i}}) + \sum_{k\in K}\psi_{k}(S_{i},w|_{S_{i}}) = v^{\alpha}(N) + \alpha_{i}v(\{i\}),
\]
and similarly for $S_{j}$.
Subtracting the two equations and using (2) yields
\[
\psi_{i}(S_{i},w|_{S_{i}}) - \psi_{j}(S_{j},w|_{S_{j}}) = \alpha_{i}v(\{i\}) - \alpha_{j}v(\{j\}).
\]
By (1), this is exactly $\psi_{i}(N,v)-\psi_{j}(N,v) = \alpha_{i}v(\{i\})-\alpha_{j}v(\{j\})$.
If $i$ and $j$ are symmetric,
then $\psi_{i}-\alpha_{i}v(\{i\})=\psi_{j}-\alpha_{j}v(\{j\})$, i.e., \textbf{$\alpha$-S}.

For $n=2$, the same conclusion follows by a direct construction with one extra player; the details are analogous to the case $n\ge 3$ (see the proof of Theorem~\ref{thm:wsp} below).
\end{proof}

\begin{theorem}\label{thm:wsp}
A value $\psi$ on $\mathcal{G}$ satisfies efficiency (\textbf{E}) and $\alpha$-subgame sum-order preservation (\textbf{$\alpha$-SSOP}) if and only if $\psi=CIS^\alpha$.
\end{theorem}

\begin{proof}
\textbf{Existence.} Direct substitution shows that $CIS^\alpha$ satisfies both axioms.

\textbf{Uniqueness.} Let $(N,v)$ be a game with $|N|=n\ge 2$ (the case $n=1$ is trivial by \textbf{E}).
If $n\ge 3$, Proposition~\ref{prop:sym} gives
\[
\psi_{i}(N,v) = \psi_{j}(N,v) + \alpha_{i}v(\{i\}) - \alpha_{j}v(\{j\})
\]
for an arbitrary  player $j\in N$.
Summing over all $j\in N$ and applying {\bf E},
\[
n\psi_{i}(N,v) = \sum_{j\in N}\psi_{j}(N,v) +n \alpha_{i}v(\{i\}) - \sum_{j\in N}\alpha_{j}v(\{j\})=v(N) +n \alpha_{i}v(\{i\}) - \sum_{j\in N}\alpha_{j}v(\{j\})
\]
which yields the weighted CIS formula for $\psi_{i}$, and therefore $\psi=CIS^\alpha$ for every player.

For $n=2$, let $N=\{i,j\}$.
Take a single new player $k$ and set $S_{i}=\{k,i\}$, $S_{j}=\{k,j\}$.
Construct $w$ on $\{i,j,k\}$ by: $w|_N=v$; $\alpha_k=0$, $w(\{k\})=0$; $w(S_{i})=v^{\alpha}(N)+\alpha_{i}v(\{i\})$, $w(S_{j})=v^{\alpha}(N)+\alpha_{j}v(\{j\})$; all other coalitions take zero.
Then $|S_{i}|=|S_{j}|=2$, $|N|=2$, and $w^{\alpha}(S_{i})=w^{\alpha}(S_{j})=v^{\alpha}(N)$.
Comparing $N$ with $S_{i}$ yields $\psi_{i}(N,v)=\psi_{i}(N,w|_N)=\psi_{i}(S_{i},w|_{S_{i}})$, and comparing $S_{i}$ with $S_{j}$ (intersection $\{k\}$) gives $\psi_{k}(S_{i},w|_{S_{i}})=\psi_{k}(S_{j},w|_{S_{j}})$.
{\bf E} on $S_{i}$ and $S_{j}$, together with these equalities, leads to the same difference formula and hence to $CIS^\alpha$.
Thus $\psi=CIS^\alpha$ for all $n$.
\end{proof}

As a byproduct of the two axioms, we obtain the $\alpha$-dummifying player property for games with at least three players. This unifies the nullifying and dummifying properties in the extreme cases.

\begin{proposition}\label{prop:fromWSP}
For all $(N,v)$ with $|N|\ge 3$, if a value $\psi$ satisfy efficiency (\textbf{E}) and $\alpha$-subgame sum-order preservation (\textbf{$\alpha$-SSOP}), then $\psi$ satisfies
$\alpha$-dummifying property (\textbf{$\alpha$-DP}).
\end{proposition}

\begin{proof}
Let $n=|N|\ge 3$ and $i\in N$ be $\alpha$-dummifying. Then $v^\alpha(S)=0$ for all $S\ni i$, so $v^\alpha(N)=0$ and $v(\{i\})=\alpha_i v(\{i\})$.

Choose a set $K$ of $n-1$ new players disjoint from $N$ and let $T=K\cup\{i\}$ ($|T|=n$). Define a game $w$ on $N\cup K$ by:
\begin{itemize}
\item $w|_N=v$, i.e., $w(S)=v(S)$ for all $S\subseteq N$;
\item For $k\in K$, set $\alpha_k=0$, and $w(R)=0$ for all $R\subseteq K$;
\item $w(\{i\})=v(\{i\})$;
\item For any $\emptyset\neq R\subseteq K$, $w(R\cup\{i\})=\alpha_i v(\{i\})$;
\item All remaining coalitions have value $0$.
\end{itemize}
One verifies $w^\alpha(S)=0$ for all $S\subseteq T$. Both $N$ and $T$ have size $n$, intersection $\{i\}$, and $v^\alpha(N)=w^\alpha(T)=0$. \textbf{$\alpha$-SSOP} therefore gives
\[
\psi_i(N,v)=\psi_i(N,w|_N)=\psi_i(T,w|_T). \tag{3}
\]

Now set $u=w|_T$ on $T=K\cup\{i\}$. Note that $u(\{i\})=v(\{i\})$, $u(\{k\})=0$ for $k\in K$, and $u(R\cup\{i\})=\alpha_i v(\{i\})$ for $R\subseteq K$, so $u^\alpha\equiv 0$ on $T$.

Fix $k\in K$. Take a disjoint set $M$ of $n-1$ new players and let $U=M\cup\{k\}$ ($|U|=n$). Define $u'$ on $T\cup M$ by:
\begin{itemize}
\item $u'|_T=u$, i.e., $u'(S)=u(S)$ for all $S\subseteq T$;
\item For $m\in M$, $\alpha_m=0$, $u'(\{m\})=0$, and all coalitions within $U$ have value $0$;
\item All other coalitions take zero.
\end{itemize}
We have $|T|=|U|=n$, $T\cap U=\{k\}$, and $u^\alpha(T)=(u')^\alpha(U)=0$. \textbf{$\alpha$-SSOP} yields
\[
\psi_k(T,u)=\psi_k(T,u'|_T)=\psi_k(U,u'|_U). \tag{4}
\]

The game $(U,u'|_U)$ is a zero game on $n\ge 3$ players with all $\alpha_m=0$. For distinct $p,q\in U$, consider $U_{-p}=U\setminus\{p\}$ and $U_{-q}=U\setminus\{q\}$; they have size $n-1$ and intersection $P=U\setminus\{p,q\}\neq\emptyset$ (since $n\ge 3$). Both have zero worth, so \textbf{$\alpha$-SSOP} forces $\sum_{j\in P}\psi_j(U_{-p},u'|_{U_{-p}})=\sum_{j\in P}\psi_j(U_{-q},u'|_{U_{-q}})$. Applying \textbf{E} to the subgames, then implies $\psi_p(U,u'|_U)=\psi_q(U,u'|_U)$. Hence all players in $U$ receive the same payoff, which must be zero by \textbf{E}. Consequently, $\psi_k(U,u'|_U)=0$, and from (4) we obtain $\psi_k(T,u)=0$.

Since $k\in K$ was arbitrary, $\sum_{k\in K}\psi_k(T,u)=0$. {\bf E} on $T$ gives $\psi_i(T,u)=u(T)=\alpha_i v(\{i\})$. Together with (3) we conclude $\psi_i(N,v)=\alpha_i v(\{i\})$.
\end{proof}

\begin{remark}
When all $\alpha_j=0$, Proposition~\ref{prop:fromWSP} states that any efficient value satisfying the corresponding subgame order preservation axiom assigns zero to nullifying players; when all $\alpha_j=1$, it assigns individual worths to dummifying players. Thus the proposition provides a unified derivation of these two classical player properties from a single order-preservation condition.
\end{remark}

\section{Concluding remarks}\label{sec:concl}

We have provided three  axiomatic characterizations of the weighted CIS value.  The first uses only generalized individual rationality and the equal surplus increment property. The second relies on efficiency, additivity, $\alpha$-symmetry, and the $\alpha$-dummifying player property. The third employs only efficiency and the $\alpha$-subgame sum-order preservation axiom. Each characterization captures the fairness and distributive rationale of the value from a distinct perspective. We also showed that efficiency together with the $\alpha$-subgame sum-order preservation axiom implies the $\alpha$-dummifying player property, which for the extreme cases all $\alpha_i=0$ and all $\alpha_i=1$ yields the nullifying and dummifying player properties, respectively. All results simultaneously cover the equal division value ($\alpha_i=0$), the standard CIS value ($\alpha_i=1$), and the $\alpha$-CIS value ($\alpha_i=\alpha$ constant).
As noted in the introduction, these characterizations unify and generalize several known results from the literature: van den Brink (2007) (\cite{vanBrink2007}), Casajus and Huettner (2014) (\cite{CasajusHuettner2014}), Navarro (2025) (\cite{Navarro2025}), and Lowing, Nakada, and Navarro (2025) (\cite{Lo2026}).


Future research may further weaken the order-preservation axioms introduced by by Lowing, Nakada, and Navarro \cite{Lo2026}. In particular, one could explore weaker forms of their ``subgame order preservation w.r.t. coalitional dual surplus ({\bf SOPCds}) and ``subgame order preservation with respect to individual contribution" ({\bf SOPIc}), with the goal of characterizing the equal allocation of non-separable contributions  (ENSC) value and the Shapley value, respectively under weaker conditions. Such investigations would further clarify the structural relations between egalitarian and marginalist solution concepts.

\section*{\bf Declaration of competing interest:} The authors declare that there is no conflict of interest.

\section*{Data availability}
No data was used for the research described in the article.

\bigskip
\noindent{\bf Acknowledgements.} 
This research is partly supported by the National Natural Science Foundation of China (72371151).

\end{document}